\documentclass[pra,twocolumn,amsmath,amssymb,showpacs, superscriptaddress,notitlepage]{revtex4-2}

\usepackage{graphicx}
\usepackage[caption=false]{subfig}
\usepackage{siunitx,booktabs}
\usepackage{float}
\usepackage{soul}
\usepackage{dcolumn}
\usepackage{amsmath,amssymb,mathtools,physics}
\usepackage{amsthm}

\theoremstyle{plain}
\newtheorem*{theorem*}{Theorem}

\newtheorem*{remark*}{Remark}

\newtheorem*{corollary*}{Corollary}

\newtheorem{lemma}{Lemma}

\usepackage{thmtools} 
\usepackage{thm-restate}

\usepackage[utf8]{inputenc}

\usepackage{tikz,pgfplots}\pgfplotsset{compat=1.18}
\usepackage{subfig}

\usepackage[hidelinks]{hyperref}
\usepackage{footnote}
\usepackage{enumitem}
\hypersetup{
	colorlinks   = true, 
	urlcolor     = blue, 
	linkcolor    = blue, 
	citecolor   = blue 
}
\usepackage[capitalize]{cleveref}

\crefname{appendix}{Appendix}{Appendices}

\usepackage [english]{babel}
\usepackage [autostyle, english = american]{csquotes} \MakeOuterQuote{"}
\usepackage{bm}
\usepackage{enumitem}

\usepackage[dvipsnames]{xcolor}

\newcommand{\nstates}{n}

\newcommand{\ndim}{n_\mathrm{dim}}

\newcommand{\nalice}{n_A}
\newcommand{\nbob}{n_B}

\newcommand{\KQKD}{\mathcal{K_\mathrm{QKD}^{\widehat{\mathcal{G}},\widehat{\mathcal{Z}}}}}

\newcommand{\gmap}{\widehat{\mathcal G}}
\newcommand{\zmap}{\widehat{\mathcal Z}}

\newcommand{\numericalcitations}{colesNumericalApproach2016,winickReliableNumerical2018,wangCharacterisingCorrelations2019,primaatmajaVersatileSecurity2019,huRobustInterior2022,zhouNumericalMethod2022,araujoQuantumKey2023,lorenteQuantumKey2025,he2024,curras-lorenzoNumericalSecurity2025,kaminRenyiSecurity2025}

\newcommand{\numericalcitationsfullchar}{colesNumericalApproach2016,winickReliableNumerical2018,wangCharacterisingCorrelations2019,primaatmajaVersatileSecurity2019,huRobustInterior2022,zhouNumericalMethod2022,araujoQuantumKey2023,lorenteQuantumKey2025,he2024}

\newcommand{\analyticalsourceimpcitations}{gottesmanSecurityQuantum2004,loSecurityQuantum2007,tamakiLosstolerantQuantum2014,pereiraQuantumKey2019,pereiraQuantumKey2020,navarretePracticalQuantum2021,pereiraModifiedBB842023,curras-lorenzoSecurityHighspeed2025}

\newcommand{\analyticalpartialcharcitations}{pereiraQuantumKey2019,pereiraQuantumKey2020,navarretePracticalQuantum2021,pereiraModifiedBB842023,curras-lorenzoSecurityHighspeed2025}

\usepackage[author=Guillermo,color=yellow]{pdfcomment}

\newcommand{\affvqcc}{Vigo Quantum Communication Center, University of Vigo, Vigo E-{36310}, Spain}
\newcommand{\affuvigo}{Escuela de Ingeniería de Telecomunicación, Department of Signal Theory and Communications, University of Vigo, Vigo E-36310, Spain}
\newcommand{\affatlantic}{atlanTTic Research Center, University of Vigo, Vigo E-36310, Spain}

\newcommand{\affmateus}{Departamento de Física Teórica, Atómica y Óptica, Laboratory for Disruptive Interdisciplinary Science (LaDIS), Universidad de Valladolid, 47011 Valladolid, Spain}

\begin{document}

	\setlength{\parskip}{3pt}
	\setlength{\parindent}{0pt}
	
	\author{Margarida Pereira} 	\author{Guillermo Currás-Lorenzo} 
	\affiliation{\affvqcc} \affiliation{\affuvigo} \affiliation{\affatlantic} 
   \author{Mateus Araújo}
	\affiliation{\affmateus} 

    	\title{Optimal key rates for quantum key distribution with partial source characterization}
	
	\setlength{\parskip}{3pt}
	\setlength{\parindent}{0pt}

\begin{abstract}
    Numerical security proofs based on conic optimization are known to deliver optimal secret-key rates, but so far they have mostly assumed that the emitted states are fully characterized. In practice, this assumption is unrealistic, since real devices inevitably suffer from imperfections and side channels that are extremely difficult to model in detail. Here, we extend conic-optimization methods to scenarios where only partial information about the emitted states is known, covering both prepare-and-measure and measurement-device-independent protocols. We demonstrate that our method outperforms state-of-the-art analytical and numerical approaches under realistic source imperfections, especially for protocols that use non-qubit encodings. These results advance numerical-based proofs towards a standard, implementation-ready framework for evaluating quantum key distribution protocols in the presence of source imperfections. 
\end{abstract}
	\maketitle

\textit{Introduction.}---In principle, quantum key distribution (QKD) can achieve information-theoretic security, a much higher standard than any other cryptosystem. However, security proofs often assume ideal device models, while real implementations inevitably suffer from imperfections. These deviations can create security loopholes that allow Eve to gain information beyond what the security proof accounts for, potentially compromising QKD's security guarantees.

For example, standard security proofs of the well-known BB84 protocol \cite{bennettQuantumCryptography1984} assume that Alice emits perfect qubit Pauli eigenstates, and that no information about her bit-and-basis choices is leaked to the eavesdropper, which is extremely difficult to satisfy in practice due to encoding imperfections and side channels. Significant progress has been made in developing analytical security proofs that incorporate such source imperfections \cite{\analyticalsourceimpcitations,wangPracticalLongDistance2019}. In particular, analytical security proofs \cite{\analyticalpartialcharcitations} have shown how to guarantee security under partial state characterization, i.e., when the emitted states are not completely known. For example, \cite{curras-lorenzoSecurityHighspeed2025} requires only the knowledge that
	\begin{equation}	\label{eq:partial_characterization}
		\ev{\rho_j}{\phi_j}
         \geq 1 - \epsilon_j,
	\end{equation}
where {$\rho_j$} is the actual state emitted by Alice when she selects setting $j$ (e.g., for BB84, $j \in \{0_Z,1_Z,0_X,1_X\}$), $\ket{\phi_j}$ is an arbitrary reference state, and $\epsilon_j$ bounds their deviation. However, these approaches often involve complicated derivations, as imperfections typically break the symmetries exploited by standard proofs, and can sometimes introduce loose bounds.

Numerical security proofs based on convex optimization \cite{\numericalcitations} have recently emerged as a powerful alternative to analytical approaches. These can often achieve optimal or near-optimal key rates even for scenarios without such symmetries, which should in principle make them well-suited for handling imperfect sources. However, they often assume a complete characterization of the emitted states \cite{\numericalcitationsfullchar}, including all imperfections. This is a very unrealistic demand, since some imperfections are extremely difficult to characterize precisely.

It is then natural to ask whether numerical security proofs could incorporate partial state knowledge as analytical proofs do. Indeed, \cite{curras-lorenzoNumericalSecurity2025} has proposed a numerical approach that can bound the asymptotic key rate for prepare-and-measure and measurement-device-independent (MDI) type protocols given the knowledge in \cref{eq:partial_characterization}. However, the approach in \cite{curras-lorenzoNumericalSecurity2025} is not fully general, and it relies on phase-error estimation, which is not optimal for some protocols \cite{matsuuraAsymptoticallyTight2025}.

Thus, it is important to consider whether partial state characterization could be integrated into more general numerical techniques that directly optimize Eve's side information on Alice's key. Here, we introduce a simple yet effective method to achieve this. Specifically, we show how to extend the optimization framework in \cite{lorenteQuantumKey2025}, which calculates \textit{optimal} QKD key rates under full state characterization, to also calculate \textit{optimal} QKD key rates under partial state characterization (i.e., the assumption in \cref{eq:partial_characterization}). Moreover, we demonstrate that our method provides secret-key rates (SKR) surpassing those in \cite{curras-lorenzoNumericalSecurity2025}, particularly for protocols employing non-qubit encodings.

\textit{Optimal key rates with full state characterization.}---Consider a general prepare-and-measure QKD protocol in which, for each round, Alice sends known states $\{\ket{\psi_j}\}_j$, selected with probabilities $\{p_j\}_j$, where $j \in \{0,\ldots,\nstates-1\}$ and $\nstates$ is the total number of emitted states. To prove security, we can use an equivalent scenario in which Alice generates the source replacement state
\begin{equation}
    \ket{\Psi}_{AA'} = \sum_{j=0}^{\nstates-1} \sqrt{p_j} \ket{j}_A \ket{\psi_j}_{A'},
\end{equation}
where $A$ is an ancillary system that Alice keeps in her possession, and $A'$ is the photonic system sent through the quantum channel. Eve's collective attack can be described as an isometry $V_{A' \to BE}$, where $B$ is the system that Eve resends to Bob, and $E$ is Eve's ancillary system containing her side information. Since we are interested in the reduced state $\rho_{AB}$ that Alice and Bob share after Eve's attack, we can define a CPTP map $\mathcal{E}_{A' \to B}$ that consists of first applying $V_{A' \to BE}$ and then tracing out system $E$. By doing so, we can write
\begin{equation}
\label{eq:eves_map}
    \rho_{AB} = (\mathcal{I} \otimes \mathcal{E}_{A' \to B}) (\ketbra{\Psi}_{AA'}).
\end{equation}
Once Bob receives system $B$, he applies a POVM $\{\Gamma_k\}_k$ to the incoming state, where $k$ denotes his measurement outcome. 

The asymptotic key rate of a QKD protocol can be written as \cite{winickReliableNumerical2018}
\begin{equation}
\label{eq:asymp_skr}
    R_\infty =  D(\mathcal{G}(\rho_{AB})\|\mathcal{Z}(\mathcal{G}(\rho_{AB}))) - p_\mathrm{pass} \lambda_\mathrm{EC}.
\end{equation}
Here, $D$ is the quantum relative entropy, $\mathcal{G}$ and $\mathcal{Z}$ are some completely positive maps defined in \cite{winickReliableNumerical2018} that specify how Alice's sifted key is extracted, $p_\mathrm{pass}$ represents the probability that a round is used for sifted key generation, and $\lambda_\mathrm{EC}$ denotes the cost of error correction per sifted key bit. Note that the latter two terms are directly observable in the protocol, but the term $D(\mathcal{G}(\rho_{AB})\|\mathcal{Z}(\mathcal{G}(\rho_{AB})))$ cannot be directly calculated since the state $\rho_{AB}$ is unknown (because Eve's map $\mathcal{E}_{A' \to B}$ is unknown). Therefore, this term needs to be optimized given the knowledge that is available to Alice and Bob. In particular, this is given by the convex optimization problem
\begin{equation}
\label{eq:convex_problem}
\begin{gathered}
    \min_{\rho_{AB}} D(\mathcal{G}(\rho_{AB})\|\mathcal{Z}(\mathcal{G}(\rho_{AB}))) \\
    \text{s.t.  } \langle i \vert \rho_{A} \vert j \rangle = \sqrt{p_i p_j} \braket{\psi_j}{\psi_i}, \quad \forall i,j \in \{0,\ldots,\nstates-1\}, \\
    \Tr[ (\dyad{j}{j}_A \otimes \Gamma_k) \rho_{AB}] = p_j Y_{k \vert j}, \quad \forall j\in \{0,\ldots,\nstates-1\}, k, \\
    \rho_{AB} \succeq 0.
\end{gathered}
\end{equation}
The first constraint in \cref{eq:convex_problem} is equivalent to
\begin{equation}
    \rho_A \coloneqq \Tr_B [\rho_{AB}] = \Tr_{A'} [\ketbra{\Psi}_{AA'}],
\end{equation}
and comes from the fact that Eve's map does not act on system $A$ (see \cref{eq:eves_map}), and therefore the marginal state on $A$ must be the same before and after Eve applies her map. Note that this implies the constraint $\Tr[\rho_{AB}] = 1$, which then should not be added explicitly, as redundant constraints lead to numerical problems. The second constraint in \cref{eq:convex_problem} comes from Alice and Bob's observations during the protocol execution, where $Y_{k \vert j}$ refers to the conditional probability that Bob obtains outcome $k$ given that Alice sends the state $\ket{\psi_j}$.

In \cite{lorenteQuantumKey2025}, it is shown that a problem of the form in \cref{eq:convex_problem} can be reformulated as the following conic optimization problem
\begin{equation}
\label{eq:conic_problem_full_characterization}
\begin{gathered}
    \min_{h,\rho_{AB}} h \\
    \text{s.t.  } \langle i \vert \rho_{A} \vert j \rangle = \sqrt{p_i p_j} \braket{\psi_j}{\psi_i}, \quad \forall i,j \in \{0,\ldots,\nstates-1\}, \\
    \text{Tr}[ (\dyad{j}{j}_A \otimes \Gamma_k) \rho_{AB}] = p_j Y_{k \vert j}, \quad \text{  }\forall j\in \{0,\ldots,\nstates-1\}, k, \\
    (h,\rho_{AB}) \in \KQKD,
\end{gathered}
\end{equation}
where $\KQKD$ is the QKD cone, defined as
\begin{equation}\label{eq:qkdcone}
\{(h,\rho) \in \mathbb R \times \mathbb H^{d};\, \rho \succeq 0, \  h \ge -H(\widehat{\mathcal G}(\rho)) + H(\widehat{\mathcal Z}(\rho)) \},
\end{equation}
where $\mathbb H^{d}$ denotes the space of $d \times d$ complex Hermitian matrices, and $\gmap, \zmap$ are obtained by performing facial reduction on $\mathcal G, \mathcal Z$ and on the constraints of the problem (see \cite{lorenteQuantumKey2025} for more information). As shown in \cite{lorenteQuantumKey2025}, a conic optimization problem of the form in \cref{eq:conic_problem_full_characterization} can be efficiently solved using interior point algorithms. In fact, the authors of \cite{lorenteQuantumKey2025} have developed a publicly available code to implement this optimization \cite{conicqkd2024}.

\textit{Optimal key rates with partial state characterization.}---The approach described above works only when Alice's emitted states are fully known. Here, we show how to extend it to the scenario in which Alice's emitted states $\{\rho_j\}_j$ are only known to satisfy \cref{eq:partial_characterization} for some known coefficients $\{\epsilon_j\}_j$, with $ 0 \leq \epsilon_j \leq 1$, and some known reference states $\{\ket{\phi_j}\}_j$.

First, we note that, given the knowledge in \cref{eq:partial_characterization}, one can assume without loss of generality that the emitted states have the form $\rho_j = \ketbra{\psi_j}$, with
\begin{equation}
\label{eq:partial_characterization_pure}
    \ket{\psi_j} = \sqrt{1-\epsilon_j} \ket{\phi_j} + \sqrt{\epsilon_j} \ket{\phi_j^\perp},
\end{equation}
where $\ket*{\phi_j^\perp}$ is an unknown state such that $\braket*{\phi_j^\perp}{\phi_j} = 0$. As proven in Appendix \ref{app:justification_pure_state_assumption} (see also \cite[Lemma 1]{curras-lorenzoNumericalSecurity2025}), any set of states $\{\rho_j\}_j$ satisfying \cref{eq:partial_characterization} can be obtained by applying a CPTP map to a set of pure states $\{\ket{\psi_j}\}_j$ of the form in \cref{eq:partial_characterization_pure}. This directly implies that, if security can be established for all such sets of pure states, then security is guaranteed for all sets of mixed states satisfying \cref{eq:partial_characterization}. Importantly, since pure states of the form in \cref{eq:partial_characterization_pure} themselves satisfy \cref{eq:partial_characterization}, this restriction does \textit{not} constitute a relaxation. Therefore, the worst-case SKR under the knowledge in \cref{eq:partial_characterization} can be expressed as the solution to the optimization problem
\begin{widetext}
\begin{equation}
\label{eq:non_conic_problem_partial_characterization}
\begin{gathered}
    \min_{h,\rho_{AB},\{\ket*{\phi_j^\perp}\}_j} h \\
    \text{s.t.  } \frac{\langle i \vert \rho_{A} \vert j\rangle}{\sqrt{p_i p_j}} = \sqrt{(1-\epsilon_i)(1-\epsilon_j)} \braket*{\phi_j}{\phi_i} + \sqrt{(1-\epsilon_i)\epsilon_j} \braket*{\phi_j^\perp}{\phi_i} + \sqrt{\epsilon_i(1-\epsilon_j)} \braket*{\phi_j}{\phi_i^\perp} + \sqrt{\epsilon_i\epsilon_j} \braket*{\phi_j^\perp}{\phi_i^\perp}, \quad \forall i,j, \\
    \text{Tr}[ (\dyad{j}{j}_A \otimes \Gamma_k) \rho_{AB}] =  p_j Y_{k \vert j}, \quad \forall j, k, \\
    \braket*{\phi_j^\perp} = 1, \quad \forall j, \\
    \braket*{\phi_j^\perp}{\phi_j} = 0, \quad \forall j, \\
    (h,\rho_{AB}) \in \KQKD.
\end{gathered}
\end{equation}
%
Next, we show how to reformulate this as a conic problem. The key observation is that the optimization problem depends on the unknown states $\{\ket*{\phi_j^\perp}\}_j$ only through their inner products with each other and with the known reference states $\{\ket{\phi_j}\}_j$. Therefore, instead of optimizing over the states themselves, we can optimize over a Gram matrix containing all these inner products, which converts the problem into a conic form. To be more precise, let $G$ be the Gram matrix of the union of vector sets $\{\ket{\phi_j}\}_j \cup \{\ket*{\phi_j^\perp}\}_j$, i.e., $G$ is the $(2\nstates \times 2\nstates)$ positive semidefinite matrix:
\begin{equation}
    G = \begin{pmatrix}
        \langle\phi_0|\phi_0\rangle & \langle\phi_0|\phi_1\rangle & \cdots & \vline & \langle\phi_0|\phi_0^\perp\rangle & \langle\phi_0|\phi_1^\perp\rangle & \cdots \\
        \langle\phi_1|\phi_0\rangle & \langle\phi_1|\phi_1\rangle & \cdots & \vline & \langle\phi_1|\phi_0^\perp\rangle & \langle\phi_1|\phi_1^\perp\rangle & \cdots \\
        \vdots & \vdots & \ddots & \vline & \vdots & \vdots & \ddots \\
        \hline
        \langle\phi_0^\perp|\phi_0\rangle & \langle\phi_0^\perp|\phi_1\rangle & \cdots & \vline & \langle\phi_0^\perp|\phi_0^\perp\rangle & \langle\phi_0^\perp|\phi_1^\perp\rangle & \cdots \\
        \langle\phi_1^\perp|\phi_0\rangle & \langle\phi_1^\perp|\phi_1\rangle & \cdots & \vline & \langle\phi_1^\perp|\phi_0^\perp\rangle & \langle\phi_1^\perp|\phi_1^\perp\rangle & \cdots \\
        \vdots & \vdots & \ddots & \vline & \vdots & \vdots & \ddots
    \end{pmatrix}.
\end{equation}
Note that some entries of this matrix are known, while others are not. More specifically, the known entries are: the entire upper left subblock (since the states $\{\ket{\phi_j}\}_j$ are known), all the diagonal entries (since all the states are normalized), and all the entries of the form $\braket*{\phi_j^\perp}{\phi_j}$ (which are equal to zero). 

This is not a relaxation either, as the existence of a Gram matrix respecting these constraints is equivalent to the existence of a set of vectors with the required inner products. Thus, by introducing $G$ as an optimization variable constrained by this known information, we can reformulate \cref{eq:non_conic_problem_partial_characterization} as the conic optimization problem
%
\begin{equation}
\label{eq:conic_problem_partial_characterization}
\begin{gathered}
    \min_{h,\rho_{AB},G} h \\
    \text{s.t.  } \frac{\langle i \vert \rho_{A} \vert j\rangle}{\sqrt{p_i p_j}} = \sqrt{(1-\epsilon_i)(1-\epsilon_j)} G_{j,i} + \sqrt{(1-\epsilon_i)\epsilon_j} G_{j+\nstates,i} + \sqrt{\epsilon_i(1-\epsilon_j)} G_{j,i+\nstates} + \sqrt{\epsilon_i\epsilon_j} G_{j+\nstates,i+\nstates}, \quad \forall i,j, \\
    \text{Tr}[ (\dyad{j}{j}_A \otimes \Gamma_k) \rho_{AB}] =  p_j Y_{k \vert j}, \quad \forall j, k, \\
    G_{i,j} = \braket{\phi_i}{\phi_j}, \quad \forall i,j; i\neq j, \\
    G_{j+\nstates,j} = 0, \quad \forall j, \\
    G_{j,j} = 1, \quad  G_{j+\nstates,j+\nstates} = 1,\quad \forall j,\\
    (h,\rho_{AB}) \in \KQKD, \\
    G \succeq 0,
\end{gathered}
\end{equation}
\end{widetext}
whose solution represents the exact SKR under partial state characterization, i.e., the knowledge in \cref{eq:partial_characterization}. Note that this is a conic optimization problem since the objective is linear and all constraints are either linear equalities or specify membership in convex cones (the QKD cone $\KQKD$ for $(h,\rho_{AB})$ and the positive semidefinite cone for $G$). Moreover, we remark that this conic optimization problem can be directly solved using the techniques in \cite{lorenteQuantumKey2025}, and in particular, one can simply reuse the publicly available code in \cite{conicqkd2024} after adding the new constraints. 

Note that when the reference states $\{\ket{\phi_j}\}_j$ are linearly dependent, no full rank $G$ exists, which can cause issues when solving \cref{eq:conic_problem_partial_characterization}. As shown in Appendix \ref{sec:lineardep}, this can be straightforwardly solved by constructing $G$ using a basis for $\text{span} \{\ket{\phi_j}\}_j$, effectively reducing its dimension.

MDI-type protocols can be handled in an analogous manner: first reformulate them in the framework of \cite{lorenteQuantumKey2025}, and then extend them with the partial state characterization constraints. We show this explicitly in Appendix \ref{app:mdi}.

\textit{Numerical results.}---Here, we apply our conic optimization approach to estimate the SKR for selected prepare-and-measure and MDI-type protocols with imperfect and partially characterized sources. To simulate the data that one would obtain in an actual implementation, we use standard channel models with detector efficiency $\eta_{d} = 0.73$, dark count probability $p_d = 10^{-6}$ and error correction inefficiency $f = 1.16$. For more information on the protocol modeling, see Appendix \ref{app:modelling}.

We compare our results with the phase-error estimation method of \cite{curras-lorenzoNumericalSecurity2025}, which provides an asymptotically optimal bound on the phase-error rate (see Appendix \ref{sec:gram} for a proof of this claim). While phase-error estimation is widely used in QKD security proofs, it is known to underestimate the SKR is some situations \cite{matsuuraAsymptoticallyTight2025}. Thus, our comparison allows us to evaluate when the phase-error approach remains near-optimal and when our direct optimization of Eve's information provides an advantage.

We first consider the BB84 protocol and assume that the reference states are of the form 
\begin{equation}
\label{eq:qubit_state_model}
    \ket{\phi_j} = \cos(\theta_j)\ket{0_Z} + \sin(\theta_j)\ket{1_Z}.
\end{equation}
Here, $\theta_j = (1+\delta/\pi) \varphi_j/2$, with $\varphi_j \in \{0,\pi,\pi/2,3\pi/2\}$ and $\delta \in [0,\pi)$, for $j \in \{0,1,2,3\}$, which corresponds to the bit-and-basis values $\{0_Z,1_Z,0_X,1_X\}$. The actual emitted states $\rho_j$ satisfy the partial characterization constraint in \cref{eq:partial_characterization} with $\epsilon_j=\epsilon~\forall j$. In this model, $\delta$ quantifies the magnitude of characterized state preparation flaws within the qubit space (e.g., polarization rotation errors), while $\epsilon$ incorporates uncharacterized imperfections such as information leakage through additional degrees of freedom or side channels \cite{curras-lorenzoSecurityHighspeed2025}. For Bob's measurement, we consider a standard BB84 setup with an active basis choice. Since numerical methods require Bob's POVM to act on a finite-dimensional system while optical detectors act on an infinite-dimensional Fock space, we apply the qubit squashing map \cite{beaudrySquashingModels2008,gittsovichSquashingModel2014}. For simplicity, we consider that all of Bob's detectors have the same efficiencies and dark count rates, which can be treated as channel-induced losses and noise, respectively \cite{naharImperfectDetectors2025}. Scenarios in which the efficiencies and dark count rates of Bob's detectors differ (in a partially characterized way) can be incorporated into our optimization problem by applying the tools in \cite{naharImperfectDetectors2025} in combination with the flag-state squasher \cite{zhangSecurityProof2021}.

Figure \ref{fig:bb84} shows the SKR as a function of distance for $\delta = 0.14$ and different values of $\epsilon$. The solid lines represent our conic optimization approach, while the dashed lines show the phase-error estimation results from \cite{curras-lorenzoNumericalSecurity2025}. We observe that both methods give nearly identical results across all parameter regimes. The phase-error approach is known to be asymptotically optimal for ideal BB84, and our results suggest it remains near-optimal even with source imperfections, at least under our imperfections model. Furthermore, since our conic optimization results coincide with those of the phase-error method, which was already shown in \cite{curras-lorenzoNumericalSecurity2025} to reproduce the analytical key rates of \cite{curras-lorenzoSecurityHighspeed2025}, these findings demonstrate that \cite{curras-lorenzoSecurityHighspeed2025} is essentially optimal for this scenario, at least in the asymptotic regime—an interesting result, since analytical bounds used to handle imperfect and partially characterized sources can often introduce looseness.
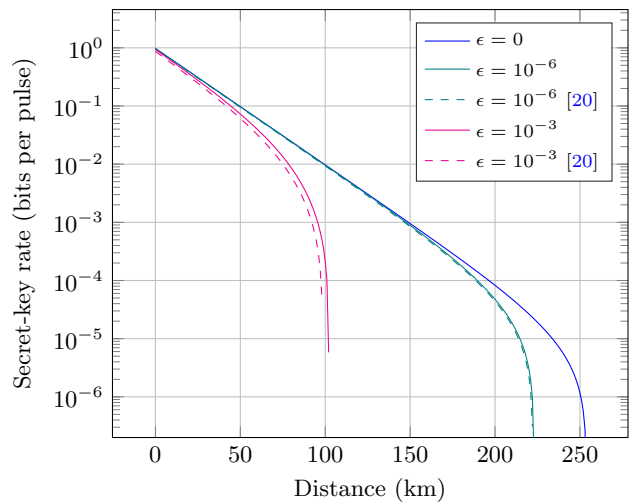
\begin{figure}[h!]
	\centering
 	\begin{tikzpicture}
		\begin{axis}[%
            ymode = log,
			ymin=2e-7,
			grid=major,
			xlabel = {Distance (km)},
            ylabel = {Secret-key rate (bits per pulse)},
			axis background/.style={fill=white},
			legend style={at={(0.97,0.97), font=\scriptsize},legend cell align=left, align=left, draw=white!15!black}
			]
            \addplot[color=blue] table[col sep=space] {plot_data/bb84e0};
            \addlegendentry{$\epsilon = 0$}
            \addplot[color=teal] table[col sep=space] {plot_data/bb84e6};
            \addlegendentry{$\epsilon = 10^{-6}$}
            \addplot[color=teal, dashed] table[col sep=space] {plot_data/bb84e6eph};
            \addlegendentry{$\epsilon = 10^{-6}$ \cite{curras-lorenzoNumericalSecurity2025}}
            \addplot[color=magenta] table[col sep=space] {plot_data/bb84e3};
            \addlegendentry{$\epsilon = 10^{-3}$}
            \addplot[color=magenta, dashed] table[col sep=space] {plot_data/bb84e3eph};
            \addlegendentry{$\epsilon = 10^{-3}$ \cite{curras-lorenzoNumericalSecurity2025}}
        \end{axis}
	\end{tikzpicture}
	\caption{Secret-key rate (logarithmic scale) as a function of distance for the BB84 protocol with $\delta = 0.14$ and different values of $\epsilon$. The solid lines correspond to our conic optimization approach while the dashed lines correspond to the phase-error estimation method of \cite{curras-lorenzoNumericalSecurity2025}. The curve obtained for $\epsilon = 0$ when using the latter analysis is omitted as its visually indistinguishable from that obtained when using our approach.}
 \label{fig:bb84}
\end{figure}  

Next, we analyze a coherent-light-based MDI protocol introduced in \cite{navarretePracticalQuantum2021}, which can be regarded as a simplified version of twin-field QKD \cite{lucamariniOvercomingRate2018}. This scheme employs a very different encoding from BB84: instead of qubit states, each user emits two coherent states with opposite phases (for key generation) plus a vacuum state (for parameter estimation). Since phase-error estimation was developed specifically for qubit-based protocols like BB84, this coherent-state encoding is a good benchmark to evaluate whether our direct optimization method provides advantages for protocols whose encoding differs significantly from BB84. 

In this case, the joint reference states of Alice and Bob are
\begin{equation}
    \ket{\phi_{ij}} = \ket{\phi_i} \otimes \ket{\phi_j}, ~~ \text{for } i,j \in \{0,1,2\},
\end{equation}
where the individual states are defined as
\begin{equation}
\label{eq:reference_states_MDI}
\begin{aligned}
    &\ket{\phi_0} = \ket{\alpha}, ~~\ket{\phi_1} = \ket{-\alpha e^{i\delta}}~~ {\rm and} ~~\ket{\phi_2} = \ket{\text{vac}},
    \end{aligned}
\end{equation}
with $\alpha$ denoting the amplitude of the coherent state and $\delta$ the characterized state preparation flaws. We consider that the joint emitted states $\rho_{ij} = \rho_i \otimes \rho_j$ satisfy the partial characterization constraint
\begin{equation}
    \ev{\rho_{ij}}{\phi_{ij}} \geq (1 - \epsilon)^2, \quad \forall i,j,
\end{equation}
where the squared term comes from the fact that this constraint could be obtained by combining two individual constraints of the form in \cref{eq:partial_characterization}, one for each user.

Figure \ref{fig:mdi} presents the SKR for this MDI-type protocol with $\delta = 0.14$ and various values of $\epsilon$. 
In this case, we observe a clear performance gap between our conic optimization (solid lines) and the phase-error approach (dashed lines) \cite{curras-lorenzoNumericalSecurity2025}. Our method consistently yields higher SKR values, and the advantage becomes more pronounced at longer distances and for smaller $\epsilon$. For instance, at $\epsilon = 0$, our approach extends the maximum achievable distance by more than 10 km compared to the phase-error estimation method. 
We have performed additional simulations with $\delta=0$ and $p_d = 10^{-4}$, and observed that, in this regime, our approach also outperforms that in \cite{curras-lorenzoNumericalSecurity2025}
(see \cref{app:additional_simulations} for more details). We note that, in the asymptotic regime, \cite{curras-lorenzoNumericalSecurity2025} already outperforms considerably the tightest analytical proof available for this protocol \cite{curras-lorenzoSecurityHighspeed2025}.
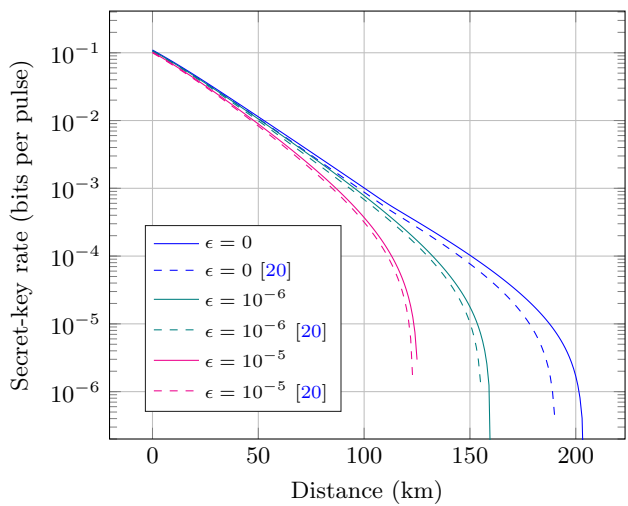
\begin{figure}[h!]
	\centering
 	\begin{tikzpicture}
		\begin{axis}[%
            ymode = log,
			ymin=2e-7,
			grid=major,
			xlabel = {Distance (km)},
            ylabel = {Secret-key rate (bits per pulse)},
			axis background/.style={fill=white},
			legend style={at={(0.45,0.5), font=\scriptsize},legend cell align=left, align=left, draw=white!15!black}
			]
            \addplot[color=blue] table[col sep=space] {plot_data/mdid14e0};
            \addlegendentry{$\epsilon = 0$}
            \addplot[color=blue, dashed] table[col sep=space] {plot_data/mdid14e0eph};
            \addlegendentry{$\epsilon = 0$ \cite{curras-lorenzoNumericalSecurity2025}}
            \addplot[color=teal] table[col sep=space] {plot_data/mdid14e6};
            \addlegendentry{$\epsilon = 10^{-6}$}
            \addplot[color=teal, dashed] table[col sep=space] {plot_data/mdid14e6eph};
            \addlegendentry{$\epsilon = 10^{-6}$ \cite{curras-lorenzoNumericalSecurity2025}}
            \addplot[color=magenta] table[col sep=space] {plot_data/mdid14e5};
            \addlegendentry{$\epsilon = 10^{-5}$}
            \addplot[color=magenta, dashed] table[col sep=space] {plot_data/mdid14e5eph};
            \addlegendentry{$\epsilon = 10^{-5}$ \cite{curras-lorenzoNumericalSecurity2025}}
        \end{axis}
	\end{tikzpicture}
	\caption{Secret-key rate rate (logarithmic scale) as a function of distance for the coherent-light-based MDI-type protocol  introduced in \cite{navarretePracticalQuantum2021}, with $\delta = 0.14$ and different values of $\epsilon$. The parameter $\alpha$ has been optimized for each distance point, taking values in the range $[0.03, 0.5]$. All computations were performed with double floats for increased precision and stability.}
 \label{fig:mdi}
\end{figure}

\textit{Conclusion.}
 In this work, we introduced a conic optimization approach that extends the framework of \cite{lorenteQuantumKey2025} from full to partial state characterization, thereby providing \textit{optimal} key rates for general prepare-and-measure and MDI-type protocols under realistic conditions. This method directly optimizes Eve’s side information and overcomes the limitations of \cite{curras-lorenzoNumericalSecurity2025}, which also considered partial characterization but relied on phase-error estimation and lacked full generality \footnote{When applied to prepare-and-measure protocols, the approach in \cite{curras-lorenzoNumericalSecurity2025} requires that Bob performs two POVMs with the same overall detection efficiency, which only makes sense in BB84-type protocols with an active measurement basis choice. Thus, the approach in \cite{curras-lorenzoNumericalSecurity2025} is not directly applicable, for example, to the six-state protocol, or even to BB84 with a passive measurement basis choice. Our approach does not have these limitations, and can be applied to general prepare-and-measure protocols.}. For qubit-based protocols such as BB84, our approach achieves nearly identical performance to the phase-error estimation method, showing that the latter is nearly optimal in this setting. In contrast, for protocols with non-qubit encodings, such as a coherent-light-based MDI protocol, our approach offers a clear improvement with respect to that in \cite{curras-lorenzoNumericalSecurity2025}. These results underscore the potential of conic optimization for analyzing realistic source imperfections and establishing security across a broad range of QKD protocols while achieving high performance.

A natural next step is applying our framework to decoy-state QKD \cite{hwangQuantumKey2003,loDecoyState2005,wangBeatingPhotonNumberSplitting2005}. For perfectly phase-randomized weak coherent sources, this is straightforward: the standard decoy analysis provides upper and lower bounds on the single-photon yields, which can be incorporated into our optimization problem by replacing the corresponding equality constraints with inequalities. However, realistic imperfections and side channels in intensity modulation and phase randomization prevent signals from being modeled as ideal photon-number mixtures, invalidating this simple approach. While \cite{kaminRenyiSecurity2025} has recently shown how to incorporate some of these imperfections into numerical-based proofs, developing a general framework that handles both bit-and-basis encoding imperfections and decoy imperfections in a partially uncharacterized manner remains an important open problem. We anticipate that the tools introduced in our work will prove valuable in addressing this challenge.

Another important challenge is the extension of our approach to the finite-key regime. For this, one could employ established methods such as the postselection technique \cite{christandlPostselectionTechnique2009,naharPostselectionTechnique2024} or the generalized entropy accumulation theorem \cite{metgerGeneralisedEntropy2022}, which would require incorporating statistical fluctuations into the constraints but would largely preserve the structure of our optimization. However, more recent approaches based on the marginal-constrained entropy accumulation theorem \cite{arqandMarginalconstrainedEntropy2025} in combination with optimization over Rényi entropies \cite{kaminRenyiSecurity2025} provide significantly tighter finite-key estimates. While adapting our framework to this setting would require extending the conic formulation to handle Rényi rather than von Neumann entropies, such an extension is expected to be feasible and would provide a promising route towards tight and practical finite-key security guarantees.

\textit{Code availability.} The code to reproduce our results is available at \url{https://github.com/araujoms/qkd_partial_characterization}.

\textit{Acknowledgements.}
We thank Matej Pivoluska for the idea to rephrase our justification in Appendix~\ref{app:justification_pure_state_assumption} in terms of source mapping and for discussions regarding the extension of our framework to decoy-state protocols. Also, we thank Marcos Curty and Álvaro Navarrete for valuable discussions. The research of M.A.~was supported by the Spanish Agencia Estatal de Investigación, Grant No.~RYC2023-044074-I, by the Q-CAYLE project, funded by the European Union-Next Generation UE/MICIU/Plan de Recuperación, Transformación y Resiliencia/Junta de Castilla y León (PRTRC17.11), and also by the Department of Education of the Junta de Castilla y León and FEDER Funds (Reference: CLU-2023-1-05). M.P. and G.C.-L.~acknowledge support from the Galician Regional Government (consolidation of research units: atlanTTic), the Spanish Ministry of Economy and Competitiveness (MINECO), the Fondo Europeo de Desarrollo Regional (FEDER) through the grant No.~PID2020-118178RB-C21, MICIN with funding from the European Union NextGenerationEU (PRTRC17.I1) and the Galician Regional Government with own funding through the “Planes Complementarios de I+D+I con las Comunidades Autonomas” in Quantum Communication, the “Hub Nacional de Excelencia en Comunicaciones Cuanticas” funded by the Spanish Ministry for Digital Transformation and the Public Service and the European Union NextGenerationEU, the European Union’s Horizon Europe Framework Programme under the Marie Sklodowska-Curie Grant No.~101072637 (Project QSI), the project “Quantum Security Networks Partnership” (QSNP, grant agreement No 101114043) and the European Union via the European Health and Digital Executive Agency (HADEA) under the Project QuTechSpace (grant 101135225). G.C.-L. acknowledges funding from the European Union's Horizon Europe research and innovation programme under the Marie Skłodowska-Curie Postdoctoral Fellowship grant agreement No.\ 101149523.

\clearpage

\appendix

\onecolumngrid

\section{Justification for \texorpdfstring{\cref{eq:partial_characterization_pure}}{Eq. (9)}}
\label{app:justification_pure_state_assumption}

Here, we provide a rigorous justification for our assumption in \cref{eq:partial_characterization_pure} of the main text. We note that this is essentially the same justification as in \cite[Lemma 1]{curras-lorenzoNumericalSecurity2025} (see also a similar previous result in \cite{curras-lorenzoSecurityHighspeed2025}) but rephrased in the language of \textit{source mapping}.

\begin{lemma}
    Let $\{\rho_j\}_j$ with $j \in \{0,1,...,\nstates-1\}$ be a set of states satisfying
    \begin{equation} \label{eqapp:partial_characterization}
        \ev{\rho_j}{\phi_j} \geq 1 - \epsilon_j, \quad\forall j,
    \end{equation}
    for a fixed set of pure states $\{\ket{\phi_j}\}_j$. Then there exist a CPTP map $\Phi$ and pure states $\{\ket{\psi_j}\}_j$ of the form
    \begin{equation}
        \label{eqapp:partial_characterization_pure}
        \ket{\psi_j} = \sqrt{1-\epsilon_j} \ket{\phi_j} + \sqrt{\epsilon_j} \ket{\phi_j^\perp},
    \end{equation}
    where $\braket{\phi_j^\perp}{\phi_j} = 0$, such that
    \begin{equation}
    \label{eq:map}
        \rho_j = \Phi\big(\ketbra{\psi_j}\big), \quad \forall j.
    \end{equation}
    As a consequence, for any prepare-and-measure QKD protocol, if security can be established for \emph{all} pure state preparations $\{\ket{\psi_j}\}_j$ satisfying \cref{eqapp:partial_characterization_pure}, then security is guaranteed for \emph{all} mixed state preparations $\{\rho_j\}_j$ satisfying \cref{eqapp:partial_characterization}.
\end{lemma}

\begin{proof}
Consider a fixed set $\{\rho_j\}_j$ satisfying \cref{eqapp:partial_characterization}. By Uhlmann's theorem, for each $j$ there exists a purification $\ket*{\psi''_j}_{aS}$ of $\rho_j$ such that
\begin{equation}
\label{eqapp:purification_assumption}
    \abs{\braket*{\phi_j}{\psi''_j}_{aS}}^2 = \ev{\rho_j}{\phi_j} \geq 1 - \epsilon_j,
\end{equation}
where $\ket{\phi_j}_{aS} \coloneqq \ket{\phi_j} \ket{0}_S$ with $S$ denoting an ancillary system. Without loss of generality, \cref{eqapp:purification_assumption} implies that $\ket*{\psi''_j}_{aS}$ can be expressed as
\begin{equation}
    \ket*{\psi''_j}_{aS} = e^{i \varphi_j}\left(\sqrt{1-\epsilon'_j} \ket{\phi_j}_{aS} + \sqrt{\epsilon'_j} \ket*{\phi_j^{\prime,\perp}}_{aS}\right),
    \label{eq:psi_phase}
\end{equation}
where $0\leq \epsilon'_j \leq \epsilon_j$, $\varphi_j \in [0,2\pi)$, and $\braket*{\phi_j^{\prime,\perp}}{\phi_j}_{aS} = 0$. Moreover, we can eliminate the global phase by defining
\begin{equation}
\label{eqapp:emitted_states_assumed_form}
    \ket*{\psi'_j}_{aS}  =e^{-i \varphi_j} \ket{\psi''_j}_{aS} = \sqrt{1-\epsilon'_j} \ket{\phi_j}_{aS} + \sqrt{\epsilon'_j} \ket*{\phi_j^{\prime,\perp}}_{aS},
\end{equation}
which is also a purification of $\rho_j$.

Next, we show that $\epsilon'_j$ can be replaced by $\epsilon_j$ by introducing a fictitious system $F$. Consider the states $\{\ket{\psi_j}_{aSF}\}_j$ defined as
\begin{equation}
\label{eqapp:psi_tilde}
    \ket{\psi_j}_{aSF} = \ket*{\psi'_j}_{aS} \otimes \left[\sqrt{\frac{1-\epsilon_j}{1-\epsilon'_j}} \ket{0}_F + \sqrt{1-\frac{1-\epsilon_j}{1-\epsilon'_j}} \ket{1}_F \right],
\end{equation}
where $\{\ket{0}_F,\ket{1}_F\}$ is an orthonormal basis for system $F$. This can be rewritten as
\begin{equation}
    \ket{\psi_j}_{aSF} = \sqrt{1-\epsilon_j} \ket{\phi_j}_{aSF} + \sqrt{\epsilon_j} \ket*{\phi_j^\perp}_{aSF},
\end{equation}
where $\ket{\phi_j}_{aSF} \coloneqq \ket{\phi_j}_{aS} \ket{0}_F$ and 
\begin{equation}
    \ket{\phi_j^\perp}_{aSF} = \frac{\sqrt{\epsilon_j - \epsilon'_j}}{\sqrt{\epsilon_j}} \ket{\phi_j}_{aS}\ket{1}_F + \frac{\sqrt{\epsilon'_j}}{\sqrt{\epsilon_j}} \ket*{\phi_j^{\prime,\perp}}_{aS} \otimes \left[\sqrt{\frac{1-\epsilon_j}{1-\epsilon'_j}} \ket{0}_F + \sqrt{\frac{\epsilon_j - \epsilon'_j}{1-\epsilon'_j}} \ket{1}_F\right]
\end{equation}
is a normalized state orthogonal to $\ket{\phi_j}_{aSF}$. 

Finally, we define the CPTP map $\Phi(\cdot) = \Tr_{SF}(\cdot)$. By construction, 
\begin{equation}
\label{eq:Phi_map}
   \rho_j = \Phi\big(\ketbra{\psi_j}_{aSF}\big),
\end{equation}
i.e., \cref{eq:map}, as required.

\textbf{Security implication.} Consider a fixed pair of sets $\{\rho_j\}_j$ and $\{\ket{\psi_j}\}_j$ related by \cref{eq:Phi_map}. Since generating $\{\rho_j\}_j$ is equivalent to first preparing $\{\ket{\psi_j}\}_j$ and then applying $\Phi$, the map $\Phi$ can be absorbed into Eve's attack. Consequently, any protocol secure against all attacks when Alice prepares $\{\ket{\psi_j}\}_j$ remains secure when she prepares $\{\rho_j\}_j$ (see, e.g., \cite[Lemma 8]{naharPostselectionTechnique2024} for a formal proof of this statement). Since we have shown that every set $\{\rho_j\}_j$ satisfying \cref{eqapp:partial_characterization} can be obtained via \cref{eq:Phi_map} from some set $\{\ket{\psi_j}\}_j$ of the form in \cref{eqapp:partial_characterization_pure}, it follows that proving security for all such pure state preparations guarantees security for all mixed state preparations satisfying \cref{eqapp:partial_characterization}.

\end{proof}

We remark that an equivalent result can also be derived for MDI-type protocols to justify our assumption in \cref{eq:partial_characterization_MDI_pure} given the knowledge in \cref{eq:partial_characterization_MDI}.

\section{Handling linearly dependent reference states}\label{sec:lineardep}
When the reference states $\{\ket{\phi_j}\}_j$ for $j \in \{0,\ldots,\nstates-1\}$, with $\nstates$ denoting the number of emitted states, are linearly dependent, the Gram matrix $G$ defined in the main text cannot have full rank, which may lead to numerical instability when solving \cref{eq:conic_problem_partial_characterization}. Here, we show how to reformulate the problem using a basis for $\text{span}\{\ket{\phi_j}\}_j$, effectively reducing the dimension of the Gram matrix.
Let $\ndim = \dim(\text{span}\{\ket{\phi_j}\}_j) \leq \nstates$ be the dimension of the subspace spanned by the reference states. We construct an orthonormal basis $\{\ket{l}\}_{l=0}^{\ndim-1}$ for this subspace, allowing us to express each reference state as
\begin{equation}
\label{eq:phij_l_expansion}
    \ket{\phi_j} = \sum_{l=0}^{\ndim-1} c_l^{(j)} \ket{l},
\end{equation}
where the coefficients $c_l^{(j)} = \braket{l}{\phi_j}$ are known since both $\{\ket{l}\}_l$ and $\{\ket{\phi_j}\}_j$ are known.
Now, instead of constructing the Gram matrix using the states $\{\ket{\phi_j}\}_j \cup \{\ket*{\phi_j^\perp}\}_j$ as in the main text, we construct a reduced Gram matrix $G$ using the basis states $\{\ket{l}\}_l \cup \{\ket*{\phi_j^\perp}\}_j$. This results in a $(\ndim + \nstates) \times (\ndim + \nstates)$ positive semidefinite matrix:
\begin{equation}
    G = \begin{pmatrix}
        \langle 0|0\rangle & \cdots & \langle 0|\ndim-1\rangle & \vline & \langle 0|\phi_0^\perp\rangle & \cdots & \langle 0|\phi_{\nstates-1}^\perp\rangle \\
        \vdots & \ddots & \vdots & \vline & \vdots & \ddots & \vdots \\
        \langle \ndim-1|0\rangle & \cdots & \langle \ndim-1|\ndim-1\rangle & \vline & \langle \ndim-1|\phi_0^\perp\rangle & \cdots & \langle \ndim-1|\phi_{\nstates-1}^\perp\rangle \\
        \hline
        \langle\phi_0^\perp|0\rangle & \cdots & \langle\phi_0^\perp|\ndim-1\rangle & \vline & \langle\phi_0^\perp|\phi_0^\perp\rangle & \cdots & \langle\phi_0^\perp|\phi_{\nstates-1}^\perp\rangle \\
        \vdots & \ddots & \vdots & \vline & \vdots & \ddots & \vdots \\
        \langle\phi_{\nstates-1}^\perp|0\rangle & \cdots & \langle\phi_{\nstates-1}^\perp|\ndim-1\rangle & \vline & \langle\phi_{\nstates-1}^\perp|\phi_0^\perp\rangle & \cdots & \langle\phi_{\nstates-1}^\perp|\phi_{\nstates-1}^\perp\rangle
    \end{pmatrix}.
\end{equation}
The known entries of this matrix are: the upper left $\ndim \times \ndim$ block, which equals $I_{\ndim}$ (since $\{\ket{l}\}_l$ is orthonormal), and the diagonal entries $G_{j+\ndim,j+\ndim} = 1$ for all $j$ (since all states are normalized). Additionally, the orthogonality constraint $\braket*{\phi_j^\perp}{\phi_j} = 0$ translates to
\begin{equation}
    \sum_{l=0}^{\ndim-1} G_{j+\ndim,l} c_l^{(j)} = 0, \quad \forall j.
\end{equation}
Using this reduced Gram matrix, we can express the required inner products as:
\begin{equation}
\begin{aligned}
    &\braket*{\phi_i}{\phi_j} = \sum_{l=0}^{\ndim-1} \overline{c_l^{(i)}} c_l^{(j)} \quad \text{(known)}, \\
    &\braket*{\phi_i^\perp}{\phi_j} = \sum_{l=0}^{\ndim-1} G_{i+\ndim,l} c_l^{(j)}, \\
    &\braket*{\phi_i}{\phi_j^\perp} = \sum_{l=0}^{\ndim-1} \overline{c_l^{(i)}} G_{l,j+\ndim}, \\
    &\braket*{\phi_i^\perp}{\phi_j^\perp} = G_{i+\ndim,j+\ndim}.
\end{aligned}
\end{equation}
The conic optimization problem then becomes
\begin{equation}
\begin{gathered}
    \min_{h,\rho_{AB},G} h \\
    \text{s.t.  } \frac{\langle i \vert \rho_{A} \vert j\rangle}{\sqrt{p_i p_j}} = \sqrt{(1-\epsilon_i)(1-\epsilon_j)} \sum_{l=0}^{\ndim-1} \overline{c_l^{(i)}} c_l^{(j)} + \sqrt{(1-\epsilon_i)\epsilon_j} \sum_{l=0}^{\ndim-1} G_{j+\ndim,l} c_l^{(i)} \\
    \quad + \sqrt{\epsilon_i(1-\epsilon_j)} \sum_{l=0}^{\ndim-1} \overline{c_l^{(j)}} G_{l,i+\ndim} + \sqrt{\epsilon_i\epsilon_j} G_{i+\ndim,j+\ndim}, \quad \forall i,j, \\
    \Tr[ (\dyad{j}{j}_A \otimes \Gamma_k) \rho_{AB}] = p_j Y_{k \vert j}, \quad \forall j, k, \\
    G_{m,m} = 1, \quad \forall m \in \{0,\ldots,\ndim+\nstates-1\}, \\
    G_{l,l'} = 0, \quad \forall l,l' \in \{0,\ldots,\ndim-1\}, l \neq l', \\
    \sum_{l=0}^{\ndim-1} G_{j+\ndim,l} c_l^{(j)} = 0, \quad \forall j, \\
    (h,\rho_{AB}) \in \KQKD \\
    G \succeq 0.
\end{gathered}
\end{equation}
We remark that an analogous approach can be applied to the optimization problem for MDI-type protocols in Appendix \ref{app:mdi} to handle the scenario in which the joint reference states $\{\ket{\phi_{ij}}\}_{i,j}$ are linearly dependent.

\section{Optimization problem for MDI-type protocols}
\label{app:mdi}
Here, we explain how to construct the optimization problem for MDI-type protocols, starting with the case of full state characterization, where Alice’s and Bob’s emitted states are completely known. We then show how to extend this construction to the scenario of partial state characterization.

\subsection{Full state characterization}
Consider a general MDI-type QKD protocol in which, for each round, Alice independently selects a setting $i \in \{0,\ldots,\nalice-1\}$ with probability $p_i^A$ and emits state $\ket*{\psi_i^A}$, while Bob independently selects a setting $j \in \{0,\ldots,\nbob-1\}$ with probability $p_j^B$ and emits state $\ket*{\psi_j^B}$. The joint quantum state sent to Eve is therefore $\ket{\psi_{ij}} = \ket*{\psi_i^A} \otimes \ket*{\psi_j^B}$, selected with probability $p_{ij} = p_i^A p_j^B$. To prove security, we employ the source replacement scheme where Alice and Bob generate the entangled state%
\begin{equation}
    \ket{\Psi}_{ABA'B'} = \sum_{i,j} \sqrt{p_{ij}} \ket{ij}_{AB} \ket{\psi_{ij}}_{A'B'},
\end{equation}
where $A$ ($B$) is an ancillary system that Alice (Bob) retains, and $A'$ ($B'$) is the photonic system sent through the quantum channel.

In MDI protocols, Eve performs a joint measurement on the incoming systems and announces a classical outcome. This can be described by an isometry $V_{A'B' \to CE}$, where $C$ is a classical register containing Eve's announcement and $E$ is Eve's quantum side information. Since we are interested in the reduced state $\rho_{ABC}$ that Alice and Bob share after Eve's announcement, we can define a CPTP map $\mathcal{E}_{A'B' \to C}$ that consists of first applying $V_{A'B' \to CE}$ and then tracing out system $E$, obtaining
\begin{equation}
\label{eq:eves_map_mdi}
    \rho_{ABC} = (\mathcal{I}_{AB} \otimes \mathcal{E}_{A'B' \to C}) (\ketbra{\Psi}_{ABA'B'}).
\end{equation}
Since $C$ is classical, we can decompose the state as
\begin{equation}\label{eq:decompositionmdi}
    \rho_{ABC} = \sum_{\gamma} \rho_{AB}^{(\gamma)} \otimes \ketbra{\gamma}_C,
\end{equation}
where $\gamma$ indexes Eve's possible announcements and $\rho_{AB}^{(\gamma)}$ represents the (subnormalized) conditional state of Alice and Bob's systems given announcement $\gamma$.

The asymptotic SKR for MDI protocols follows \cref{eq:asymp_skr} from the main text, but now involves optimizing over the family of subnormalized states $\{\rho_{AB}^{(\gamma)}\}_{\gamma}$:
\begin{equation}
\label{eq:convex_problem_mdi}
\begin{gathered}
    \min_{\{\rho_{AB}^{(\gamma)}\}_{\gamma}} D(\mathcal{G}(\rho_{ABC})\|\mathcal{Z}(\mathcal{G}(\rho_{ABC}))) \\
    \text{s.t.  } \rho_{ABC} = \sum_{\gamma} \rho_{AB}^{(\gamma)} \otimes \ketbra{\gamma}_C, \\
    \sum_{\gamma}\langle ij \vert \rho_{AB}^{(\gamma)} \vert i'j' \rangle = \sqrt{p_{ij} p_{i'j'}} \braket{\psi_{i'j'}}{\psi_{ij}}, \quad \forall i,i',j,j', \\
    \text{Tr}[ \ketbra{ij}_{AB} \rho_{AB}^{(\gamma)}] = p_{ij} Y_{\gamma \vert ij}, \quad \forall i,j,\gamma, \\
    \rho_{AB}^{(\gamma)} \succeq 0, \quad \forall \gamma.
\end{gathered}
\end{equation}
The second constraint is equivalent to
\begin{equation}
   \rho_{AB} \coloneqq \text{Tr}_C[\rho_{ABC}] = \sum_{\gamma} \rho_{AB}^{(\gamma)} =\text{Tr}_{A'B'}  [\ketbra{\Psi}_{ABA'B'}],
\end{equation}
and comes from the fact that Eve's map does not act on systems $AB$ (see \cref{eq:eves_map_mdi}), and therefore the marginal state on $AB$ must be the same before and after Eve applies her map. The third constraint in \cref{eq:convex_problem_mdi} comes from Alice and Bob's observations during protocol execution, where $Y_{\gamma|ij}$ refers to the conditional probability that Eve announces outcome $\gamma$ given that Alice selects setting $i$ and Bob selects setting $j$. The remaining constraints ensure that $\rho_{ABC}$ is properly normalized and positive semidefinite.

As in the main text, following the approach in \cite{lorenteQuantumKey2025}, this can be reformulated as the conic optimization problem:
\begin{equation}
\label{eq:conic_problem_mdi}
\begin{gathered}
    \min_{h,\,\{\rho_{AB}^{(\gamma)}\}_{\gamma}} h \\
    \text{s.t.  } \rho_{ABC} = \sum_{\gamma} \rho_{AB}^{(\gamma)} \otimes \ketbra{\gamma}_C, \\
    \sum_{\gamma}\langle ij \vert \rho_{AB}^{(\gamma)} \vert i'j' \rangle = \sqrt{p_{ij} p_{i'j'}} \braket{\psi_{i'j'}}{\psi_{ij}}, \quad \forall i,i',j,j', \\    \text{Tr}[ \ketbra{ij}_{AB} \rho_{AB}^{(\gamma)}] = p_{ij} Y_{\gamma \vert ij}, \quad \forall i,j,\gamma, \\
    (h, \rho_{ABC}) \in \KQKD.
\end{gathered}
\end{equation}

\subsection{Partial state characterization}

Now consider the realistic scenario where Alice and Bob only have partial knowledge of their emitted states. Specifically, they only know that their joint state $\rho_{ij} = \rho_i \otimes \rho_j$, emitted when they select settings $i$ and $j$, respectively, satisfies
\begin{equation}
    \ev{\rho_{ij}}{\phi_{ij}} \geq  1 - \varepsilon_{ij},
    \label{eq:partial_characterization_MDI}
\end{equation}
where $\{\ket{\phi_{ij}}\}_{ij}$ are known reference states and $\{\varepsilon_{ij}\}_{ij}$ are known deviation bounds s.t.~$0\leq \varepsilon_{ij} \leq 1$. Following \cref{app:justification_pure_state_assumption}, we can assume, without loss of generality, that the joint emitted states are pure, i.e.\ $\rho_{ij} = \ketbra{\psi_{ij}}$, and have the form
\begin{equation}
\label{eq:partial_characterization_MDI_pure}
    \ket{\psi_{ij}} = \sqrt{1-\varepsilon_{ij}} \ket{\phi_{ij}} + \sqrt{\varepsilon_{ij}} \ket*{\phi_{ij}^{\perp}},
\end{equation}
where $\ket*{\phi_{ij}^{\perp}}$ are unknown states orthogonal to the reference states: $\braket*{\phi_{ij}^{\perp}}{\phi_{ij}} = 0$.

The key insight from the main text, i.e., that the optimization depends only on inner products between states, applies directly to the MDI case as well. We introduce a Gram matrix $G$ containing all inner products between the states in $\{\ket{\phi_{ij}}\}_{ij} \cup \{\ket*{\phi_{ij}^{\perp}}\}_{ij}$. This is a $(2\nalice\nbob \times 2\nalice\nbob)$ positive semidefinite matrix:
\begin{equation}
   G = \begin{pmatrix}
       \langle\phi_{00}|\phi_{00}\rangle & \langle\phi_{00}|\phi_{01}\rangle & \cdots & \vline & \langle\phi_{00}|\phi_{00}^\perp\rangle & \langle\phi_{00}|\phi_{01}^\perp\rangle & \cdots \\
       \langle\phi_{01}|\phi_{00}\rangle & \langle\phi_{01}|\phi_{01}\rangle & \cdots & \vline & \langle\phi_{01}|\phi_{00}^\perp\rangle & \langle\phi_{01}|\phi_{01}^\perp\rangle & \cdots \\
       \vdots & \vdots & \ddots & \vline & \vdots & \vdots & \ddots \\
       \hline
       \langle\phi_{00}^\perp|\phi_{00}\rangle & \langle\phi_{00}^\perp|\phi_{01}\rangle & \cdots & \vline & \langle\phi_{00}^\perp|\phi_{00}^\perp\rangle & \langle\phi_{00}^\perp|\phi_{01}^\perp\rangle & \cdots \\
       \langle\phi_{01}^\perp|\phi_{00}\rangle & \langle\phi_{01}^\perp|\phi_{01}\rangle & \cdots & \vline & \langle\phi_{01}^\perp|\phi_{00}^\perp\rangle & \langle\phi_{01}^\perp|\phi_{01}^\perp\rangle & \cdots \\
       \vdots & \vdots & \ddots & \vline & \vdots & \vdots & \ddots
   \end{pmatrix},
\end{equation}
where each index pair $(ij)$ for each subblock runs over all $\nalice \times \nbob$ combinations of Alice's and Bob's settings. Note that some entries of $G$ are known, while others are unknown. More specifically, the known entries are: the entire upper left subblock (since the states $\{\ket{\phi_{ij}}\}_{ij}$ are known), all the diagonal entries (since all the states are normalized), and all the entries of the form $\braket*{\phi_{ij}^{\perp}}{\phi_{ij}}$ (which are equal to zero). Thus, by introducing $G$ as an optimization variable constrained by this known information, we can reformulate the optimization problem with partial state characterization as a conic problem. For MDI protocols with partial state characterization, this problem becomes
\begin{equation}
\label{eq:conic_problem_mdi_partial}
\begin{gathered}
    \min_{h,\, \{\rho_{AB}^{(\gamma)}\}_{\gamma},\, G} h \\
    \text{s.t.  } 
    \rho_{ABC} = \sum_{\gamma} \rho_{AB}^{(\gamma)} \otimes \ketbra{\gamma}_C, \\
    \frac{\langle ij \vert \rho_{AB}^{(\gamma)} \vert i'j' \rangle}{\sqrt{p_{ij} p_{i'j'}}} = \sum_{a,b \in \{0,1\}} \sqrt{\varepsilon_{ij}^a(1-\varepsilon_{ij})^{1-a}\varepsilon_{i'j'}^b(1-\varepsilon_{i'j'})^{1-b}} G_{(ij,a),(i'j',b)}, \\
    \text{Tr}[ \ketbra{ij}_{AB} \rho_{AB}^{(\gamma)}] = p_{ij} Y_{\gamma \vert ij}, \quad \forall i,j,\gamma, \\
    G_{(ij,0),(ij,0)} = 1, \quad G_{(ij,1),(ij,1)} = 1, \quad G_{(ij,0),(ij,1)} = 0, \quad  \forall i,j, \\
    G_{(ij,0),(i'j',0)} = \braket{\phi_{ij}}{\phi_{i'j'}}, \quad \forall (ij) \neq (i'j'), \\
    (h, \rho_{ABC}) \in \KQKD, \quad G \succeq 0.
\end{gathered}
\end{equation}
where, for notational convenience, we have indexed the matrix $G$ using pairs $(ij,a)$ for $a \in \{0,1\}$, with $a=0$ corresponding to $\ket{\phi_{ij}}$ and $a=1$ to $\ket*{\phi_{ij}^{\perp}}$.

\section{Protocol modeling}\label{app:modelling}

Here, we explain how to model and simulate the SKR for the BB84 protocol and the coherent-light-based MDI protocol introduced in \cite{navarretePracticalQuantum2021}, which are considered in the numerical simulations presented in the main text.

\subsection{BB84 protocol}

Since this is a prepare-and-measure protocol, we use the conic optimization problem in \cref{eq:conic_problem_partial_characterization}. In this protocol, Alice sends four different states indexed by $j \in \{0,1,2,3\}$, which correspond to the bit-and-basis values $\{0_Z,1_Z,0_X,1_X\}$, respectively, and are selected with probabilities
\begin{equation}
\label{eq:bb84_probabilities}
    p_0 = p_1 = {p_{Z_A}}/{2} \qquad \text{and} \qquad p_2=p_3={p_{X_A}}/{2},
\end{equation}
with $p_{X_A} = 1-p_{Z_A}$. The reference states $\ket{\phi_j}$ are given by \cref{eq:qubit_state_model}, from which all the inner products $\braket{\phi_{j'}}{\phi_{j}}$ can be directly calculated.

As for Bob, we consider that he performs an active basis choice, although our approach can equally be applied to setups in which he performs a passive basis choice, e.g., using a 50:50 beam splitter to divide the incoming signal. In an active basis choice setup, Bob chooses between the $Z$ basis or the $X$ basis measurement with probability $p_{Z_B}$ and $p_{X_B}=1-p_{Z_B}$, respectively, and performs the selected measurement, obtaining either a bit value or an inconclusive outcome. Since the focus of our work is on source imperfections, we assume that Bob's detectors are identical, and that the only non-idealities of his measurement setup are: (1) non-unity detector efficiencies and (2) dark counts. For attempts to incorporate partially characterized imperfections in Bob's setup into numerical security proofs such as ours, we refer the reader to Ref.~\cite{naharImperfectDetectors2025}.

Bob's $Z$ and $X$ POVMs can be described by $\{\Gamma_0^{(Z)},\Gamma_1^{(Z)},\Gamma_\bot\}$ and $\{\Gamma_0^{(X)},\Gamma_1^{(X)},\Gamma_\bot\}$, respectively. However, in our model, Bob performs a single POVM $\{\Gamma_k\}_k$. This is not a problem, since we can regard Bob's overall action as a five outcome POVM $\{\Gamma_{0Z},\Gamma_{1Z},\Gamma_{0X},\Gamma_{1X},\Gamma_\bot\}$, where
\begin{equation}
  \Gamma_{0Z} = p_{Z_B} \Gamma_0^{(Z)}, \quad \Gamma_{1Z} = p_{Z_B} \Gamma_1^{(Z)}, \quad \Gamma_{0X} = p_{X_B} \Gamma_0^{(X)}, \quad \Gamma_{1X} = p_{X_B} \Gamma_1^{(X)} {\rm ~~and~~~}
  \Gamma_{\perp} = \mathbb{I} - \Gamma_{0Z} - \Gamma_{1Z} - \Gamma_{0X} - \Gamma_{1X}.
\end{equation}

Note that, in reality, Bob's POVM acts on an infinitely-dimensional system, since Eve can resend to Bob any number of photons she wishes. However, for an active BB84 setup in which Bob's detectors are identical, we can apply the qubit squasher \cite{beaudrySquashingModels2008,gittsovichSquashingModel2014}. Moreover, we can consider without loss of generality that any non-unit efficiency in Bob's detectors is part of the quantum channel, and their dark counts can be regarded as channel-induced noise (see \cite{naharImperfectDetectors2025} for a proof of this). Thus, we can consider that Bob's POVM is ideal and acts on a three-dimensional system with orthonormal basis $\{\ket{0}_B,\ket{1}_B,\ket{2}_B\}$, where $\ket{0}_B$ ($\ket{1}_B$) corresponds to a single-photon with horizontal (vertical) encoding, and $\ket{2}_B$ corresponds to no photon. In particular, we can consider Bob's POVM elements to be
\begin{equation}
\label{eq:bb84_bob_povm_ideal}
    \Gamma_{0Z} = p_{Z_B} \ketbra{0}_B, \quad  \Gamma_{1Z} = p_{Z_B} \ketbra{1}_B, \quad \Gamma_{0X} = p_{X_B} \ketbra{+}_B, \quad \Gamma_{1X} = p_{X_B} \ketbra{-}_B {\rm ~~and~~~} \Gamma_{\bot} = \ketbra{2}_B,
\end{equation}
where we have defined $\ket{\pm}_B = (\ket{0}_B \pm \ket{1}_B)/\sqrt{2}$.

Now, we define the maps $\mathcal{G}$ and $\mathcal{Z}$, which specify the filtered state from which the sifted key is extracted, and how Alice extracts her sifted key bits, respectively. Here, we consider the efficient version of the BB84 protocol, in which Alice and Bob use the $Z$ basis most of the time, and extract the key only from this basis. This  maximizes the sifting efficiency, and is thus expected to result in higher SKRs. We remark, however, that our technique also directly applies to the case in which Alice and Bob extract key from both bases, for any values of $p_{Z_A}$ and $p_{Z_B}$.

Since the sifted key is extracted from the events in which Alice selects $j=0$ or $j=1$ and Bob obtains $k=0Z$ or $k=1Z$, the map $\mathcal{G}$ has the form
\begin{equation}
    \mathcal{G} (\rho) = G \rho G,
\end{equation}
where
\begin{equation}\label{eq:bb84g}
    G = \sqrt{(\ketbra{0}_A + \ketbra{1}_A) \otimes (\Gamma_{0Z}+\Gamma_{1Z})} = \sqrt{p_{Z_B}} (\ketbra{0}_A + \ketbra{1}_A) \otimes (\ketbra{0}_B + \ketbra{1}_B).
\end{equation}
On the other hand, since Alice assigns bit "0" when $j=0$ and bit "1" when $j=1$, the map $\mathcal{Z}$ has the form
\begin{equation}
    \mathcal{Z}(\sigma) = Z_0 \sigma Z_0 + Z_1 \sigma Z_1,
\end{equation}
where
\begin{equation}
    Z_0 = \ketbra{0}_A \otimes (\ketbra{0}_B + \ketbra{1}_B) {\rm ~~and~~} Z_1 = \ketbra{1}_A \otimes (\ketbra{0}_B + \ketbra{1}_B).
\end{equation}

As explained in Ref.~\cite{lorenteQuantumKey2025}, in order to use those maps in the conic problem in \cref{eq:conic_problem_partial_characterization}, we need to perform facial reduction on them, that is, find strictly positive maps $\gmap, \zmap$ such that $\mathcal{G}(\rho) = V\gmap(\rho)V^\dagger$ and $\mathcal{Z} \circ \mathcal{G}(\rho) = W \zmap(\rho) W^\dagger$ for some isometries $V$ and $W$. A straightforward calculation shows that the following maps satisfy this: $\gmap(\rho) := V^\dagger\mathcal{G}(\rho)V$ and $\zmap(\rho) :=W^\dagger(\mathcal{Z} \circ \mathcal{G}(\rho))W$ for $V = W = (\ket{0}_A\bra{0}_{A'} + \ket{1}_A\bra{1}_{A'}) \otimes (\ket{0}_B\bra{0}_{B'} + \ket{1}_B\bra{1}_{B'})$, where $\{\ket{0}_{A'}, \ket{1}_{A'}\}$ and $\{\ket{0}_{B'}, \ket{1}_{B'}\}$ are orthonormal bases of dimension 2, instead of the original dimensions 4 and 3, respectively.

The only remaining undefined terms in \cref{eq:conic_problem_partial_characterization} and the key rate formula in \cref{eq:asymp_skr} are the yield probabilities $Y_{k|j}$, the probability that a round is used for key generation 
\begin{equation}\label{eq:ppassBB84}
p_\mathrm{pass} = \frac{p_{Z_A}}{2} (Y_{Z0\vert 0}+Y_{Z0\vert 1}+Y_{Z1\vert 0}+Y_{Z1\vert 1}),
\end{equation}
and the cost of error correction per sifted key bit $\lambda_\mathrm{EC}$. In a real implementation, these terms would be observed during the protocol execution. Instead, here we estimate them using the channel model below. 

\textit{Note}: Although the approach described in this section is valid for any $p_{Z_A}, p_{Z_B} \in (0,1)$ the asymptotic key rate is achieved in the limit $p_{Z_A}, p_{Z_B} \to 1$. While this limit can be directly implemented in \cref{eq:bb84g,eq:ppassBB84} by setting $p_{Z_A} = p_{Z_B} = 1$, setting $p_{X_A} = p_{X_B} = 0$ in \cref{eq:bb84_probabilities,eq:bb84_bob_povm_ideal,eq:conic_problem_partial_characterization} leads to divisions by zero and loss of constraints. To circumvent this issue, we simultaneously set $p_{Z_A} = p_{Z_B} = 1$ and $p_{X_A} = p_{X_B} = 1$ throughout the optimization. While this results in an unphysical scenario with $\Tr[\rho_{AB}]=2$, the optimization problem remains well-defined and correctly recovers the asymptotic key rate. This approach was used to generate the results in the main text. \\

\noindent \textbf{Channel model} \\
We consider a standard BB84 channel model with a lossy channel and an active-basis receiver. The overall transmission efficiency is $\eta=\eta_c\eta_d$ with channel transmission $\eta_c=10^{-0.02\,l}$ (for channel length $l$ in km) and detector efficiency $\eta_d$. The dark-count probability is denoted by $p_d$, and for simplicity we assume no channel misalignment. By neglecting $O(p_d^2)$ terms, we can express the yields as
\begin{equation}
\begin{aligned}
&Y_{Z0|j} = p_{Z_B} \left( (1-\eta)\,p_d \;+\; \frac{\eta}{2}\Big[\,1 + \big(1-p_d\big)\cos\!\big(2\theta_j\big)\,\Big]\right), \\
&Y_{Z1|j} =  p_{Z_B} \left((1-\eta)\,p_d \;+\; \frac{\eta}{2}\Big[\,1 - \big(1-p_d\big)\cos\!\big(2\theta_j\big)\,\Big]\right), \\
&Y_{X0|j} =  p_{X_B} \left((1-\eta)\,p_d \;+\; \frac{\eta}{2}\Big[\,1 + \big(1-p_d\big)\sin\!\big(2\theta_j\big)\,\Big]\right), \\
&Y_{X1|j} =  p_{X_B} \left((1-\eta)\,p_d \;+\; \frac{\eta}{2}\Big[\,1 - \big(1-p_d\big)\sin\!\big(2\theta_j\big)\,\Big]\right),
\label{eq:yieldsBB84}
\end{aligned}
\end{equation}
where $Y_{B\beta\,|\,j}$ denotes the probability that Bob selects basis $B \in \{Z,X\}$ and obtains bit outcome $\beta \in \{0,1\}$, conditioned on Alice having selected setting $j$. As discussed in the main text, we set $\theta_j = (1+\delta/\pi)\,\varphi_j/2$, with $\varphi_j \in \{0,\pi,\pi/2,3\pi/2\}$ corresponding to $j \in \{0,1,2,3\}$, where $\delta \in [0,\pi)$ quantifies the magnitude of the state preparation flaws.

As for the cost of error correction, we assume that $\lambda_\mathrm{EC} = f h(e_Z)$, where $f$ is the error correction inefficiency, $h(x) = -x \log_2 x - (1-x) \log_2 (1-x)$ is the binary entropy function, and $e_Z$ is the bit-error rate, which is calculated from the yields in \cref{eq:yieldsBB84} as
\begin{align}
e_Z = \frac{Y_{Z1|0} + Y_{Z0|1}}
{Y_{Z0|0} + Y_{Z1|0} + Y_{Z0|1} + Y_{Z1|1}}. \label{eq:eZdef}
\end{align}

\subsection{Coherent-light-based MDI protocol}

In the coherent-light-based MDI protocol in \cite{navarretePracticalQuantum2021}, Alice and Bob ideally prepare the states $\{\ket{\alpha}, \ket{-\alpha}, \ket{\rm vac}\}$, where $\ket{\alpha}$ and $\ket{-\alpha}$ are coherent states with opposite phases (used for key generation), and $\ket{\rm vac}$ is a vacuum state (used only for parameter estimation). We model characterized state preparation flaws by considering that one of the emitted coherent states is shifted by an angle $\delta$. That is, for each user, we consider the reference states
\begin{equation}
\begin{aligned}
    &\ket{\phi_0} = \ket{\alpha}, ~~\ket{\phi_1} = \ket{-\alpha e^{i\delta}}~~ {\rm and} ~~\ket{\phi_2} = \ket{\text{vac}},
    \end{aligned}
\end{equation}
with the joint reference states given by $\ket{\phi_{ij}} = \ket{\phi_i} \otimes \ket{\phi_j}$. The inner products $\braket{\phi_{i'j'}}{\phi_{ij}}$ can then be calculated directly. We consider that the state emission probabilities are $p_0 = p_1 = p^c/2$ and $p_2 = 1-p^c$, where $p^c$ is the probability that the users select a coherent state and each coherent state is selected equally likely. The joint selection probabilities are thus $p_{ij} = p_i p_j$ for $i,j \in \{0,1,2\}$.

If Charlie is honest, he interferes the two incoming signals at a $50{:}50$ beam splitter, followed by threshold detectors $D_c$ and $D_d$ placed at the output ports corresponding to constructive and destructive interference, respectively. He then announces the result $\gamma=0$ if only detector $D_c$ clicks, $\gamma=1$ if only detector $D_d$ clicks, or $\gamma=2$ if neither detector clicks or both detectors click. Of course, Charlie may be dishonest, and even controlled by Eve, so they may select the announcements using any measurement and strategy they see fit. 

Note that, in the MDI case, no squashing map is needed since Alice and Bob's systems $A$ and $B$ are both guaranteed to be finite dimensional, since they are fictitious ancillary systems kept in their labs.

Alice and Bob extract their sifted key from the rounds in which $i,j\in\{0,1\}$ and $\gamma \in \{0,1\}$, with $i=0$ or $j=0$ corresponding to bit 0 and $i=1$ or $j=1$ corresponding to bit 1. For the rounds in which $\gamma=1$, Bob flips his sifted key bit, since destructive interference should indicate anticorrelation between Alice's and Bob's coherent states. The maps $\mathcal{G}$ and $\mathcal{Z}$ are thus given by
\begin{subequations}
\begin{gather}
    \mathcal{G}(\rho) = G\rho G,~ \textrm{with}~G = (\ketbra{0}_A+\ketbra{1}_A) \otimes (\ketbra{0}_B+\ketbra{1}_B) \otimes (\ketbra{0}_C+\ketbra{1}_C), \\
    \mathcal{Z}(\sigma) = Z_0 \sigma Z_0 + Z_1 \sigma Z_1, ~ \textrm{with}~ Z_j = \ketbra{j}_A \otimes (\ketbra{0}_B+\ketbra{1}_B) \otimes (\ketbra{0}_C+\ketbra{1}_C) \textrm{ for } j \in \{0,1\}.
\end{gather}
\end{subequations}
For a more efficient optimization, we now use the decomposition \eqref{eq:decompositionmdi} to simplify these maps. Applying them to $\rho_{ABC}$ we get
\begin{subequations}
\begin{gather}
    \mathcal{G}(\rho_{ABC}) = \sum_{\gamma=0}^1 \mathcal{G}'(\rho_{AB}^{(\gamma)}) \otimes \ketbra{\gamma}_C, \\
    \mathcal{Z} \circ \mathcal{G}(\rho_{ABC}) = \sum_{\gamma=0}^1 \mathcal{Z}'(\rho_{AB}^{(\gamma)})\otimes \ketbra{\gamma}_C,
\end{gather}
\end{subequations}
where
\begin{subequations}
\begin{gather}
    \mathcal{G}'(\rho) = G'\rho G',~ \textrm{with}~G' = (\ketbra{0}_A+\ketbra{1}_A) \otimes (\ketbra{0}_B+\ketbra{1}_B), \\
    \mathcal{Z}'(\sigma) = Z_0' \sigma Z_0' + Z_1' \sigma Z_1', ~ \textrm{with}~ Z_j' = \ketbra{j}_A \otimes (\ketbra{0}_B+\ketbra{1}_B) \textrm{ for } j \in \{0,1\}.
\end{gather}
\end{subequations}
Now, we use the fact that $D(\sum_i \rho_i \otimes \ketbra{i}\| \sum_i \sigma_i \otimes \ketbra{i}) = \sum_i D(\rho_i \| \sigma_i)$ and conclude that
\begin{equation}
    D(\mathcal{G}(\rho_{ABC}) \| \mathcal{Z} \circ \mathcal{G}(\rho_{ABC})) = \sum_{\gamma=0}^1 D(\mathcal{G}'(\rho_{AB}^{(\gamma)}) \| \mathcal{Z}'(\rho_{AB}^{(\gamma)})),
\end{equation}
which allows us to break down the relative entropy optimization into two cones of lower dimension. As in the BB84 case, we now need to do facial reduction on the maps $\mathcal{G}', \mathcal{Z}'$. Let then $V = (\ket{0}_A\bra{0}_{A'} + \ket{1}_A\bra{1}_{A'}) \otimes (\ket{0}_B\bra{0}_{B'} + \ket{1}_B\bra{1}_{B'})$, where $\{\ket{0}_{A'}, \ket{1}_{A'}\}$ and $\{\ket{0}_{B'}, \ket{1}_{B'}\}$ are orthonormal bases of dimension 2, instead of the original 3. Then $\gmap(\rho) := V^\dagger \mathcal{G}'(\rho)V$ and $\zmap(\rho) := V^\dagger \mathcal{Z}'(\rho)V$ are strictly positive and respect $\mathcal{G}'(\rho) = V \gmap(\rho) V^\dagger$ and $\mathcal{Z}'(\rho) = V \zmap(\rho) V^\dagger$, as required.

The remaining terms in the optimization program in \cref{eq:conic_problem_mdi_partial} and the key rate formula in \cref{eq:asymp_skr} are the yield probabilities $Y_{\gamma|ij}$, the probability that a round is used for key generation
\begin{equation}
\label{eq:ppass_MDI}
  p_\mathrm{pass} = \sum_{i,j\in \{0,1\}} p_{ij} (Y_{0\vert ij} + Y_{1 \vert ij}),  
\end{equation}
and the cost of error correction per sifted key bit $\lambda_\mathrm{EC}$. In a real implementation, these terms would be observed during the protocol execution. Instead, here we estimate them using the channel model below. 

\textit{Note}: Similarly to the BB84 protocol, the approach described in this section is valid for any $p^c \in (0,1)$. However, the asymptotic key rate is achieved in the limit $p^c \to 1$ and $p_2=1-p^c \to 0$. While this limit can be directly implemented in \cref{eq:ppass_MDI}, setting $p_2=0$ in \cref{eq:conic_problem_mdi_partial} leads to divisions by zero and loss of constraints. To circumvent this issue, one can simultaneously set $p^c = 1$ and $p_2 = 1$ throughout the optimization. While this results in an unphysical scenario with $\Tr[\rho_{ABC}]=4$, the optimization problem remains well-defined and correctly recovers the asymptotic key rate. Equivalently, one can set $p^c = 2/3$ and $p_2  = 1/3$ and multiply the result by $9/4$. We used the latter approach to generate the results in the main text and in \cref{app:additional_simulations}. \\

\noindent \textbf{Channel model} \\
We model the links from Alice and Bob to Charlie as independent lossy channels with transmittance $\eta$. For simplicity, we assume that both detectors at Charlie's have a dark-count probability $p_d$. Recall that Charlie announces $\gamma = 0$ if detector $D_c$ clicks and detector $D_d$ does not click, $\gamma = 1$ if detector $D_c$ does not click and detector $D_d$ clicks, and $\gamma = 2$ otherwise. Under these assumptions, the yields $Y_{\gamma \vert ij}$ can be expressed as
\begin{equation}
\begin{aligned}
Y_{0|i,j}
&= \Big[1-\exp\!\Big(-\frac{\eta}{2}\,\big|\phi_i+\phi_j\big|^2\Big)\,(1-p_d)\Big]\,
   \exp\!\Big(-\frac{\eta}{2}\,\big|\phi_i-\phi_j\big|^2\Big)\,(1-p_d), \\
Y_{1|i,j}
&= \Big[1-\exp\!\Big(-\frac{\eta}{2}\,\big|\phi_i-\phi_j\big|^2\Big)\,(1-p_d)\Big]\,
   \exp\!\Big(-\frac{\eta}{2}\,\big|\phi_i+\phi_j\big|^2\Big)\,(1-p_d),
   \label{eq:yields_MDI}
\end{aligned}
\end{equation}
and $Y_{2|i,j}=1-Y_{0|i,j}-Y_{1|i,j}$. The parameters $\phi_i$ in \cref{eq:yields_MDI} represent the complex amplitudes of the coherent states, with values $\phi_0=\alpha$, $\phi_1=-\alpha e^{i\delta}$, and $\phi_2=0$, where $\delta \in [0,\pi)$ quantifies the magnitude of the state preparation flaws.

As for the cost of error correction, we assume that $\lambda_\mathrm{EC} = f h(e_\textrm{bit})$, where $f$ is the error correction inefficiency, $h(x) = -x \log_2 x - (1-x) \log_2 (1-x)$ is the binary entropy function, and $e_\textrm{bit}$ is the bit-error rate, which can be defined directly from the yields in \cref{eq:yields_MDI} as
\begin{equation}
e_\textrm{bit} =
\frac{\displaystyle
\sum_{(i,j)\in\mathcal{S}_{\mathrm{diff}}} Y_{0|i,j}
\;+\;\sum_{(i,j)\in\mathcal{S}_{\mathrm{same}}} Y_{1|i,j}
}
{\displaystyle
\sum_{(i,j)\in\mathcal{S}_{\mathrm{same}}\cup\mathcal{S}_{\mathrm{diff}}}
\Big( Y_{0|i,j} + Y_{1|i,j} \Big)}, \label{eq:mdi-ez}
\end{equation}
where $\mathcal{S}_{\mathrm{same}}=\{(0,0),(1,1)\}$ and $\mathcal{S}_{\mathrm{diff}}=\{(0,1),(1,0)\}$.

\section{Reconstructing solution}\label{sec:gram}
Here, we show that the Gram matrix formulation of the phase-error rate estimation in \cite{curras-lorenzoNumericalSecurity2025} is not a relaxation, but an equivalent formulation. We do this by showing that given a positive semidefinite Gram matrix $G$ respecting the constraints, we can reconstruct measurement operators $M_\gamma$ and states $\ket{\phi_j}, \ket*{\phi_j^\perp}$ that reproduce the entries of $G$.

Since $G$ is positive semidefinite, there exists a matrix $V$ such that $V^\dagger V = G$. The columns of $V$ are the vectors $\ket{M_\gamma\phi_j},\ket*{M_\gamma\phi_j^\perp}$ that give the desired inner products. We need then to find $M_\gamma, \ket{\phi_j}, \ket*{\phi_j^\perp}$ such that $M_\gamma\ket{\phi_j} = \ket{M_\gamma\phi_j}$ and $M_\gamma\ket*{\phi_j^\perp} = \ket*{M_\gamma\phi_j^\perp}$.

Let then $g_\gamma$ be the principal submatrix of $G$ containing the rows and columns corresponding to $\ket{M_\gamma\phi_j},\ket*{M_\gamma\phi_j^\perp}$ for all $j$. Since they are principal submatrices of a positive semidefinite matrix, they are also positive semidefinite. Now let $g = \sum_\gamma g_\gamma$. It is also positive semidefinite, and therefore there exists $W$ such that $W^\dagger W = G$. Let $\ket{\phi_j}, \ket*{\phi_j^\perp}$ be the columns of $W$. We have that $\braket*{\phi_j^{(\perp)}}{\phi_i^{(\perp)}} = \sum_\gamma \braket*{M_\gamma\phi_j^{(\perp)}}{M_\gamma\phi_i^{(\perp)}}$, and therefore the inner products between the reconstructed vectors obey the constraints imposed on $G$. 

It remains to reconstruct the measurement operators $M_\gamma$. This amounts to solving the linear systems $M_\gamma\ket{\phi_j} = \ket{M_\gamma\phi_j}$ and $M_\gamma\ket*{\phi_j^\perp} = \ket*{M_\gamma\phi_j^\perp}$ over $M_\gamma$, or equivalently $M_\gamma W = V_\gamma$, where $V_\gamma$ contains the columns of $V$ associated to index $\gamma$.
To do this we need a more careful construction of $W$. Let an eigendecomposition of $g$ be $\sum_j \lambda_j \ketbra{v_j}{v_j}$, and define 
\begin{equation}
    W = \sum_{j;\lambda_j > 0} \sqrt{\lambda_j}\ketbra{j}{v_j},
\end{equation}
where $\{\ket{j}\}$ is an orthonormal basis of dimension $\operatorname{rank}(g)$. This particular choice of $W$ has a right inverse
\begin{equation}
W^R = \sum_{j;\lambda_j > 0} \frac1{\sqrt{\lambda_j}}\ketbra{v_j}{j},
\end{equation}
and for all $\gamma$ we can define $M_\gamma = V_\gamma W^R$.

\section{Additional simulations}
\label{app:additional_simulations}

Here, we investigate the SKR of the coherent-light-based MDI protocol when $\delta=0$ and $p_d = 10^{-4}$. The results shown in \cref{fig:mdipd4} indicate that, in this regime, there is also a clear difference between the two approaches, with our method outperforming that in \cite{curras-lorenzoNumericalSecurity2025}.
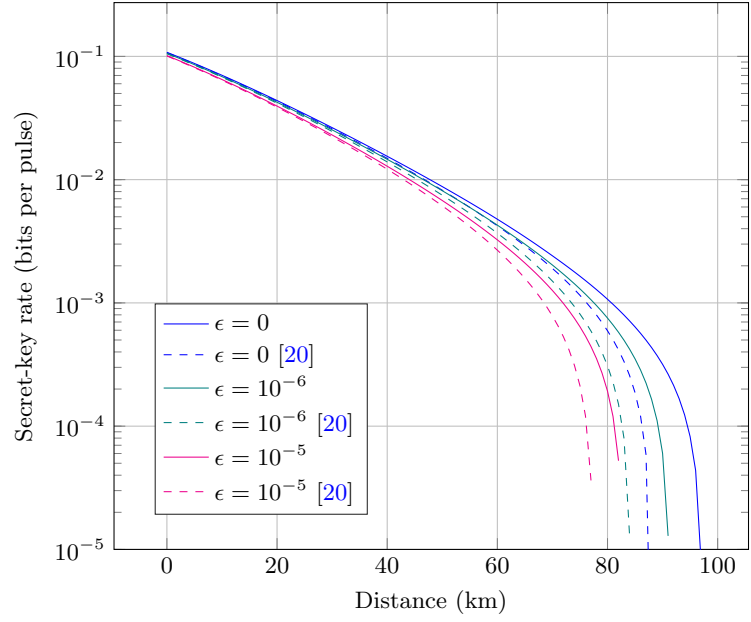
\begin{figure}[h!]
	\centering
 	\begin{tikzpicture}
		\begin{axis}[%
			scale only axis,
            ymode = log,
			ymin=1e-5,
			grid=major,
			xlabel = {Distance (km)},
            ylabel = {Secret-key rate (bits per pulse)},
			axis background/.style={fill=white},
			legend style={at={(0.4,0.45)},legend cell align=left, align=left, draw=white!15!black}
			]
            \addplot[color=blue] table[col sep=space] {plot_data/mdid0e0pd4};
            \addlegendentry{$\epsilon = 0$}
            \addplot[color=blue, dashed] table[col sep=space] {plot_data/mdid0e0pd4eph};
            \addlegendentry{$\epsilon = 0$ \cite{curras-lorenzoNumericalSecurity2025}}
            \addplot[color=teal] table[col sep=space] {plot_data/mdid0e6pd4};
            \addlegendentry{$\epsilon = 10^{-6}$}
            \addplot[color=teal, dashed] table[col sep=space] {plot_data/mdid0e6pd4eph};
            \addlegendentry{$\epsilon = 10^{-6}$ \cite{curras-lorenzoNumericalSecurity2025}}
            \addplot[color=magenta] table[col sep=space] {plot_data/mdid0e5pd4};
            \addlegendentry{$\epsilon = 10^{-5}$}
            \addplot[color=magenta, dashed] table[col sep=space] {plot_data/mdid0e5pd4eph};
            \addlegendentry{$\epsilon = 10^{-5}$ \cite{curras-lorenzoNumericalSecurity2025}}
        \end{axis}
	\end{tikzpicture}
	\caption{Secret-key rate (logarithmic scale) as a function of distance for the coherent-light-based MDI protocol with $\delta = 0$, $p_d = {10^{-4}}$ and different values of $\epsilon$ (which accounts for both Alice's and Bob's source imperfections). The solid lines correspond to our conic optimization approach while the dashed lines correspond to the phase-estimation method of \cite{curras-lorenzoNumericalSecurity2025}. The parameter $\alpha$ has been optimized for each distance point, taking values roughly in the range $[0.1, 0.4]$. All computations were performed with double floats for increased precision and stability.}
 \label{fig:mdipd4}
\end{figure} 

\end{document}